\documentclass[runningheads]{llncs}

\usepackage{amsmath}

\usepackage{stmaryrd} 
\usepackage{amsmath}
\usepackage{amssymb}
\usepackage[table]{xcolor}
\usepackage{xspace}
\usepackage{graphicx}
\usepackage{complexity}
\usepackage{bm}
\usepackage{tikz}
\usetikzlibrary{petri,positioning, fit, shapes, calc}
\usepackage{microtype}

\usepackage{thm-restate}

\usepackage{subcaption}
\newcommand{\defeq}{\stackrel{\scriptscriptstyle\text{def}}{=}}

\newcommand{\ie}{\text{i.e.}\xspace}

\newcommand{\etal}{\text{et al.}\xspace} 

\newcommand{\N}{\mathbb{N}}                    
\newcommand{\set}[1]{\left\{#1\right\}}        

\newcommand{\multiset}[1]{\Lbag#1\Rbag}        

\newcommand{\preset}[1]{{}^\bullet #1}  
\newcommand{\postset}[1]{{#1}^\bullet}  

\newcommand{\net}{N}

\newcommand{\trans}[1]{\xrightarrow{#1}}       
\newcommand{\targetMarking}{M}
\newcommand{\sourceMarking}{M'}

\newcommand{\Rset}{A} 
\newcommand{\restrictions}{\mathcal{R}}

\title{Efficient Restrictions of Immediate Observation Petri Nets\thanks{This project has received funding from the European Research Council (ERC) under the European Union's Horizon 2020 research and innovation programme under grant agreement No 787367 (PaVeS)}}
\titlerunning{Efficient restrictions of IO nets}

\author{Michael Raskin \and Chana Weil-Kennedy
\email{\{raskin,chana.weilkennedy\}@in.tum.de}
\\ ORCID: 0000-0002-6660-5673, 0000-0002-1351-8824
}
\institute{Technical University of Munich, Germany}
\authorrunning{M.\,Raskin \and C.\,Weil-Kennedy}
\begin{document}

\maketitle

\begin{abstract}
In a previous paper we introduced
immediate observation Petri nets \cite{EsparzaRW19}, a subclass of Petri nets with application domains in distributed protocols and theoretical chemistry (chemical reaction networks). 
IO nets enjoy many useful properties
\cite{EsparzaRW19,flatness-io-bio},
but like the general case of conservative Petri nets 
they have a \PSPACE-complete reachability problem.
In this paper we explore two restrictions of the reachability problem for IO nets
which lower the complexity of the problem drastically.
The complexity is \NP-complete for the first restriction 
with applications in distributed protocols, 
and it is polynomial
for the second restriction with applications in chemical settings.
%
%
        \keywords{Petri nets, reachability, computational complexity}
\end{abstract}

\section{Introduction}

In this paper we refine our results about the complexity 
of verifying immediate observation Petri nets \cite{EsparzaRW19}
in the case of two restrictions of such nets.
Petri nets and their subclasses 
are widely used and studied in the
context of software and system verification (e.g. \cite{DBLP:journals/automatica/DavidA94}),
but also others such as 
game theory (e.g. \cite{DBLP:conf/apn/Lin08}),
chemical reaction networks (e.g. \cite{DBLP:journals/nc/BaldanCMS10})
etc.
Unfortunately many important problems there
have high complexity,
and reachability is at least $\mathsf{TOWER}$-hard in the general case \cite{DBLP:conf/stoc/CzerwinskiLLLM19}.
This motivates the study of subclasses of Petri nets.

Immediate observation Petri nets (IO nets)
are a reformulation of
immediate observation population protocols,
which have been introduced by
Angluin \etal in \cite{journals/dc/AngluinAER07}.
Initially, they were studied from the point of view
of computing predicates in a distributed system,
where their
expressive power is lower than 
general population protocols 
(conservative Petri nets)
but still considerable.
Many verification problems for IO nets
are \PSPACE-complete;
among them set-parametrized problems for  sets
defined by boolean combinations of bounds
on token counts.
This is a significant improvement compared
to the general or conservative case  of Petri nets,
where
$\mathsf{EXPSPACE}$-hard \cite{DBLP:conf/stoc/CardozaLM76}
and even harder
verification problems are the norm.
IO nets provide a natural description of some distributed systems,
but also can be used to describe enzymatic chemical networks \cite{angeli2007petri}.

Of course,
a subclass of reachability problems 
with a better computational complexity
raises some natural, even if informal, questions.
What allows better complexity and can it be generalized
to some wider subclass?
What keeps the complexity from being even lower
and are there useful subclasses without these obstacles?
Are there applications where a typical problem
can be solved more efficiently?
We believe that branching immediate observation nets,
a generalization of IO nets and basic parallel processes
with reachability problem in \PSPACE\cite{flatness-io-bio},
answer the first question.
The present paper is devoted to the last two questions.

We consider two restrictions, the first one a syntactic restriction defining a subclass of IO nets,
 and the second a condition on the initial and final markings 
considered in the reachability problem for IO nets. 
The first restriction is plausible in some distributed systems, and 
it also bears similarity to
the delayed observation population protocols
introduced by Angluin et al. in \cite{journals/dc/AngluinAER07}.
The second restriction has
 applications in some chemical systems (enzymatic chemical reaction networks,  \cite{angeli2007petri}).
We show the first restriction entails an NP-complete reachability problem, 
and for the second restriction we provide a polynomial algorithm  deciding reachability or giving a witness that the restriction does not hold.

The rest of the paper is organized as follows.
In section \ref{sec:preliminaries}, we recall some general definitions regarding Petri nets,
as well as the classic maximum flow minimum cut problem.
Section \ref{sec:io-nets} defines immediate observation Petri nets.
Then we show the effects for reachability complexity
of two
restrictions on IO nets: keeping transitions enabled once enabled in Section \ref{sec:restriction1},
and requiring all token counts and their combinations to be large or zero in Section \ref{sec:restriction2}.
Finally, we summarize our results in the conclusion
and outline some further directions.

\section{Preliminaries}
\label{sec:preliminaries}

\medskip \noindent \textbf{Multisets.} 
A \emph{multiset} on a finite set \(E\) is a mapping \(C \colon E \rightarrow \N\), i.e. for any $e\in E$, \(C(e)\) denotes the number of occurrences of element \(e\) in \(C\).
Let $\multiset{e_1,\ldots,e_n}$ denote the multiset $C$ such that $C(e)=|\{j\mid e_j=e\}|$.
Operations on \(\N\) like addition or comparison are extended to multisets by defining them component wise on each element of \(E\).
Given \(X\subseteq E\) define \(C(X)\defeq\sum_{e\in X} C(e)\).
We call $\sum_{e\in E} C(e)$ the \emph{size} of $C$ and note it $|C|$.
%

\medskip \noindent \textbf{Place/transition Petri nets with weighted arcs.}
A \emph{Petri net} $N$ is a triple $(P,T,W)$ consisting of a finite set of \emph{places} $P$, a finite set of \emph{transitions} $T$ and a \emph{weight function} $W \colon (P \times T) \cup (T \times P) \rightarrow \mathbb{N}$. 
A \emph{marking} $M$ is a multiset on $P$, and we say that a marking $M$ puts $M(p)$ \emph{tokens} in place $p$ of $P$. The \emph{size} of $M$, denoted by $|M|$, is the total number of tokens in $M$.
The \emph{preset} $\preset{t}$ and \emph{postset} $\postset{t}$ of a transition $t$ of $T$ are the multisets on $P$ given by $\preset{t}(p)=W(p,t)$ and $\postset{t}(p)=W(t,p)$. A transition $t$ is \emph{enabled} at a marking $M$ if $\preset{t} \leq M$, i.e. $\preset{t}$ is component-wise smaller or equal to $M$.
If $t$ is enabled then it can be \emph{fired}, leading to a new marking $M'=M - \preset{t} + \postset{t}$. 
We let $M \xrightarrow{t} M'$ denote this.
%
Given $\sigma=t_1 \ldots t_n$ we write $M \xrightarrow{\sigma} M_n$ when $M \xrightarrow{t_1} M_1 \xrightarrow{t_2} M_2 \ldots \xrightarrow{t_n} M_n$, and call $\sigma$ a \emph{firing sequence}. 
We write $M' \trans{*} M''$ if $M' \xrightarrow{\sigma} M''$ for some $\sigma \in T^*$, and say that $M''$ is \emph{reachable} from $M'$. 
%

\medskip \noindent \textbf{Flows and cuts.} 
A \emph{flow graph} is a triple $G=(V,A,c)$ 
where $V$ is a finite set of vertices,
$A \subseteq V^2$ is a finite set of arcs,
and $c:A\rightarrow \N \cup \set{\infty}$ is a nonnegative 
\emph{capacity}
function
on arcs.
Given an arc $a\in A$, we call $c(a)$ the \emph{capacity} of $a$.
Notice that this capacity can be infinite.
A flow graph contains two special vertices $i$ and $o$, 
called the \emph{inlet} and \emph{outlet}, 
such that $i$ has no incoming arc and $o$ has no outgoing arc.
A \emph{flow} of a flow graph is a function $f:A \rightarrow \N$ 
such that $f(a)\le c(a)$ for each arc $a\in A$,
and for each vertex $v\in V \setminus \set{i,o}$, 
the sum of the flow over $v$'s incoming arcs
is equal to the sum of the flow over $v$'s outgoing arcs.
The \emph{value} of a flow is the 
sum $\sum_{(i,p)\in A} f((i,p))$ of the flow over all arcs from the inlet,
or equivalently the 
sum $\sum_{(p,o)\in A} f((p,o))$ of the flow over all arcs to the outlet.
A \emph{cut} in a flow graph $G=(V,A,c)$
is a pair of disjoint subsets $V_I \sqcup V_O = V$ such that 
the inlet is in $V_I$ and the outlet is in $V_O$.
The \emph{capacity} of a cut $(V_I, V_O)$ is the sum
of the capacities of all the arcs 
going from vertices in $V_I$ to vertices in $V_O$.
We say an arc $a=(u,v)$ \emph{crosses the cut}, if $u\in V_i$ and $v \in V_O$.

We recall two classic theorems.

\begin{theorem}[Max-flow min-cut theorem \cite{ford_fulkerson_1956}]
\label{thm:FordFulkerson}
In a flow graph, the maximum value of a flow is equal to 
the minimum capacity of a cut.
\end{theorem}

\begin{theorem}[Dinitz algorithm \cite{dinic1970}]
Given a flow graph, a flow with the maximum value and a cut with the minimum capacity
can be found in polynomial time.
\end{theorem}

\section{Immediate observation Petri nets}
\label{sec:io-nets}
We recall the definition of immediate observation nets (IO nets) from \cite{EsparzaRW19}.

\begin{definition} 
\label{def:IOnet}
A transition $t$ of a Petri net is an \emph{immediate observation transition} (IO transition)  if there are 
places $p_s, p_d, p_o$, not necessarily distinct, such that $\preset{t}=\multiset{p_s,p_o}$ and $\postset{t}=\multiset{p_d,p_o}$.
We call $p_s, p_d, p_o$ the \emph{source}, \emph{destination}, and \emph{observed} places of $t$, respectively. 
We denote by $p_s \trans{p_o} p_d$ such a transition.
A Petri net is an \emph{immediate observation net} (IO net) 
if all its transitions are IO transitions.
\end{definition}

Following the graphical convention of \cite{MontanariR95} for contextual nets, 
we represent the Petri net arcs $(p_o,t)$ and $(t,p_o)$ by an undirected arc between $t$ and $p_o$ in our figures.
This emphasizes that transition $t$ has a read-only relation to its observed place $p_o$.
In the examples, we also consider IO nets containing transitions with no observed place. 
To make the net a formally correct IO net, it suffices to add an extra marked place which acts as observed place for these transitions.

IO nets are \emph{conservative}, \ie there is no creation or destruction of tokens.

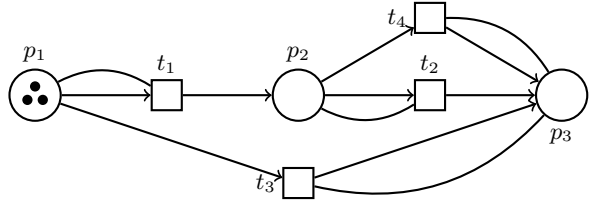
\begin{figure}[h]
\centering%
\begin{tikzpicture}[->, node distance=1.75cm, auto, thick]
      \node[place, tokens=3] (p1) {};
      \node[transition] (t1) [right of=p1] {};
      \node[place] (p2) [right of=t1] {};
      \node[transition] (t3) [below=0.6 of p2] {};
      \node[transition] (t2) [right of=p2] {};
      \node[place] (p3) [right of=t2] {};
      \node[transition] (t4) [above=0.6 of t2] {};
      
      \path[-]
      (p1) edge[bend left] node {} (t1)
      (t2) edge[bend left] node {} (p2)
      (t3) edge[bend right] node {} (p3)
      (t4) edge[bend left] node {} (p3)
      ;
      \path[->]
      (p1) edge node {} (t1)
      (t1) edge node {} (p2)
      (p2) edge node {} (t2)     
      (t2) edge node {} (p3)
      (p1) edge node {} (t3)
      (p2) edge node {} (t4)
      (t4) edge node {} (p3)
      (t3) edge node {} (p3)
      ;

      \node[] () [above= -1pt of p1] {$p_1$};
      \node[] () [above= -1pt of p2] {$p_2$};
      \node[] () [below= -1pt of p3] {$p_3$};
      \node[] () [above= -1pt of t1] {$t_1$};
      \node[] () [above= -1pt of t2] {$t_2$};
      \node[] () [left= -1pt of t3] {$t_3$};
      \node[] () [left= -1pt of t4] {$t_4$};
\end{tikzpicture}
\caption{An IO net.}%
\label{fig:io}
\vskip-5mm
\end{figure}

\begin{example}
Figure \ref{fig:io} shows  an IO net taken from the literature on population protocols \cite{journals/dc/AngluinAER07}. 
Intuitively, it models a protocol allowing a crowd of undistinguishable agents that can only interact in pairs 
to decide whether they are at least 3.
Given a marking $M_0$ with tokens only in $p_1$, if $M_0(p_1) \geq 3$,
then repeated firing of an arbitrary enabled transition
eventually puts all the tokens into $p_3$.
\end{example}

In \cite{EsparzaRW19}, 
we showed that given an IO net $\net$ and two markings $M, M'$,
deciding whether $M'$ is reachable from $M$ is a \PSPACE-complete problem.
The proof of \PSPACE-hardness for the reachability problem in IO nets
uses a reduction from the halting problem
of linear-space Turing machines.
The reduction is done 
by simulating the 
runs of the Turing machine: 
places describe the state of the head and of the tape cells, 
and transitions model the movement of the head and the change in the symbols on the tape cells.
In the construction a specific ``success'' place becomes marked
if and only if the machine reaches the halting state
without exceeding the permitted space.
 
The nets provided by this reduction have
two common properties.
First, the transitions get enabled and disabled a  large number of times.
Second, the markings put at most one token per place.
We show how forcing a strong enough contrary condition to 
at least one of these properties leads to much easier
verification.


\section{First restriction: transition enabling}
\label{sec:restriction1}

The \PSPACE-hardness proof for IO reachability relies
on the observation requirements of some transitions
switching between satisfied and unsatisfied many times.
In some distributed systems, observations correspond
to irrevocable declarations of the agents,
for example in some multi-phase commit protocols.
We consider IO nets where a token move enabled by observing some token remains enabled even when the observed token has changed places.
We formalize such a property in the following definition.

\begin{definition}
        An IO net is \emph{non-forgetting} if
        for each transitions 
        $p\xrightarrow{r}q$ and
        $r\xrightarrow{s}r'$
        there is also a transition 
        $p\xrightarrow{r'}q$.
\end{definition}

Consider a marking of an IO net
where the observation place of some transition with source place $p$ and destination place $q$ is marked.
If there is a token in place $p$, then it can move to $q$.
We say that the \emph{token move from $p$ to $q$ is enabled}.
In a non-forgetting IO net, once the token move from $p$ to $q$ is enabled in some marking of a firing sequence,
it stays enabled in  the subsequent markings of the firing sequence.
Notice that the token move from $p$ to $q$ being enabled in a marking
is not equivalent to a transition from $p$ to $q$ being enabled:
a transition is enabled when both its observation place and its source place are marked,
whereas a token move is enabled as soon as the observation place of some suitable transition is marked.

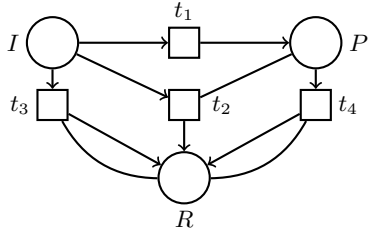
\begin{figure}
\vskip-5mm
        \centerline{\mbox{
\begin{tikzpicture}[->, node distance=1.75cm, auto, thick]

        \node[place] (I) {};
        \node[transition] (verify) [right of=I] {};
        \node[place] (V) [right of=verify] {};
        \node[transition] (object) [below=4mm of verify] {};
        \node[transition] (join) [left of=object] {};
        \node[transition] (concede) [right of=object] {};
        \node[place] (B) [below=4mm of object] {};

        \path[-]
        (V) edge (object) 
        (B) edge [bend right] (concede) 
        (B) edge [bend left] (join)
        ;
        \path[->]
        (I) edge (verify) (verify) edge (V)
        (I) edge (object) (object) edge (B)
        (I) edge (join)  
        (join) edge node {} (B)
        (V) edge (concede) 
        (concede) edge node {} (B)
        ; 
        
      \node[] () [left= -1pt of I] {$I$};
      \node[] () [right= -1pt of V] {$P$};
      \node[] () [below= -1pt of B] {$R$};
      \node[] () [above= -1pt of verify] {$t_1$};
      \node[] () [right= +1pt of object] {$t_2$};
      \node[] () [left= -1pt of join] {$t_3$};
      \node[] () [right= -1pt of concede] {$t_4$};
\end{tikzpicture} }}

        \caption{
A non-forgetting Petri net.
        }
        \label{fig:non-forgetting-net}
\vskip-5mm
\end{figure}

\begin{example}
\label{ex:non-forgetting}
The non-forgetting IO net
of Figure \ref{fig:non-forgetting-net} models
one of the steps of updating a shared state: A proposal can be published and stored,
and every agent has an opportunity to veto it.

All agents start in the initial state $I$.
Some agent can propose a change by moving from state $I$ to state $P$.
If there is a proposal, an agent can move from state $I$ to state $P$
to support the proposal,
or go to the state $R$ to reject the proposal.
If there is an agent rejecting the proposal (i.e. in the state $R$),
other agents can move to $R$ both from $I$ and from $P$ to
recognise the fact that the proposal has been rejected.
Note that the agents cannot reject a proposal before it has been
created, which is encoded by $P$ being the observed place of $t_2$.
Also note that the agent proposing a change cannot start rejecting
it until some other agent rejects it.
\end{example}

The reachability problem for such IO nets becomes much simpler.

\begin{theorem}
The reachability problem for non-forgetting IO nets is in \NP.
\end{theorem}

\begin{proof}
Let $\net$ be a non-forgetting IO net.
Consider a (non-empty) firing sequence $\sigma$ of $\net$ from markings $M$ to $M'$.
It can be decomposed into $n$ non-empty subsequences 
$\sigma_i$ such that $M=M_0\trans{\sigma_1} M_1\trans{\sigma_2} M_2\ldots \trans{\sigma_n}M_n=M'$ for some $n>0$, and
such that $M_i$ are the markings of the firing sequence
in which new token moves become enabled. 
Recall that since $\net$ is non-forgetting, 
a token move once enabled remains enabled. 
There are at most $|P|^2$ such subsequences in any firing sequence, and 
in each subsequence the set of enabled token moves is fixed.
\begin{example}
Consider the net of Example \ref{ex:non-forgetting}, 
and the firing sequence $t_2^3 t_4$ from marking $(4,1,0)$, 
which put 4 tokens in $I$, 1 tokens in $P$ and 0 token in $R$, 
to marking $(1,0,4)$. 
This firing sequence is decomposed into two subsequences:
 $(4,1,0) \trans{t_2} (3,1,1)$ and $ (3,1,1)\trans{t_2^2 t_4}(1,0,4)$.
 In the first, the token moves from $I$ to $R$ and from $I$ to $P$ are enabled.
 In the second, these token moves as well as the token move from $P$ to $R$ are enabled.

\end{example}
To show that
	the reachability problem
	for non-forgetting IO nets
	is in \NP,
we define a reachability certificate
and show how it can be verified in polynomial time.
	The certificate 
	corresponding to a firing sequence
	consists of the markings
	in which some token move is enabled for the first time.
	Such a certificate has polynomial length 
by the above considerations on the number of subsequences.

We now show that the reachability problem in an IO net with a fixed set of enabled token moves is reducible to the maximum flow problem on graphs.
	Let $\net$ be an IO net, let $M,M'$ be two markings of $\net$.
	We define $G$ as the flow graph 
	with vertices identified with the places $P$ of $\net$,
	as well as two additional vertices $i$ and $o$, the inlet and outlet
	of the flow graph.
	For each enabled token move from $p$ to $q$ for some places $p,q$,
	there is an arc from $p$ to $q$ in $G$ with infinite capacity.
	Each vertex $p$ identified with a place of $\net$
	has one incoming arc from the inlet $i$
	with capacity $M(p)$,
	and one outgoing arc to the outlet $o$
	with capacity $M'(p)$.

\begin{example}	
	Figure \ref{fig:non-forgetting-flow} illustrates 
	two such flow graphs 
        for the non-forgetting IO net of Example \ref{ex:non-forgetting}.
        The first flow graph corresponds to the enabled token moves 
        from $I$ to $R$ and from $I$ to $P$, with markings $M=(4,1,0)$ and $M'=(3,1,1)$.
        The second flow graph corresponds to the enabled token moves 
        from $I$ to $R$, from $I$ to $P$ and from $P$ to $R$, with markings $M=(3,1,1)$ and $M'=(1,0,4)$.
        \end{example}

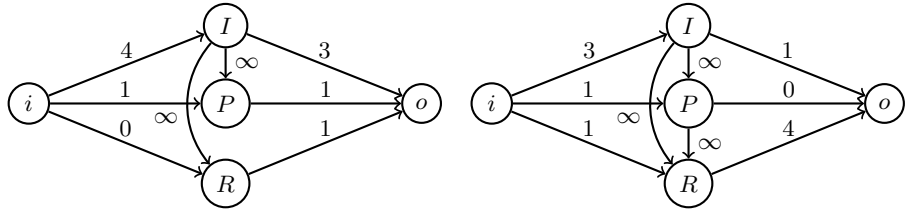
\begin{figure}
\vskip-5mm
	\mbox{

\begin{tikzpicture}[->, node distance=1.5cm, auto, thick]

\node[circle,draw] (inlet) {$i$};
\node[circle,draw] (V) [right=2 of inlet] {$P$};
\node[circle,draw] (I) [above=4mm of V] {$I$};
\node[circle,draw] (B) [below=4mm of V] {$R$};
\node[circle,draw] (outlet) [right=2cm of V] {$o$};

        \path[->]
        (inlet) edge node  [above=-0.3mm] {4} (I)
        (inlet) edge node  [above=-0.4mm] {1} (V)
        (inlet) edge node  [above=-0.3mm] {0} (B)
        (I) edge node {$\infty$} (V)
        (I) edge[bend right=40] node[below left] {$\infty$} (B)
        (I) edge node [above=-0.3mm] {3} (outlet)
        (V) edge node [above=-0.4mm] {1} (outlet)
        (B) edge node [above=-0.3mm] {1} (outlet)
        ;

\end{tikzpicture} }
        \hfil
        \mbox{

\begin{tikzpicture}[->, node distance=1.5cm, auto, thick]

\node[circle,draw] (inlet) {$i$};
\node[circle,draw] (V) [right=2 of inlet] {$P$};
\node[circle,draw] (I) [above=4mm of V] {$I$};
\node[circle,draw] (B) [below=4mm of V] {$R$};
\node[circle,draw] (outlet) [right=2cm of V] {$o$};

        \path[->]
        (inlet) edge node  [above=-0.3mm] {3} (I)
        (inlet) edge node  [above=-0.4mm] {1} (V)
        (inlet) edge node  [above=-0.3mm] {1} (B)
        (I) edge node {$\infty$} (V)
        (V) edge node {$\infty$} (B)
        (I) edge[bend right=40] node[below left] {$\infty$} (B)
        (I) edge node [above=-0.3mm] {1} (outlet)
        (V) edge node [above=-0.4mm] {0} (outlet)
        (B) edge node [above=-0.3mm] {4} (outlet)
        ;

\end{tikzpicture} }
        \caption{Flow graphs corresponding to the non-forgetting net of 
        Fig.~\ref{fig:non-forgetting-net}.}
        \label{fig:non-forgetting-flow}
\vskip-5mm
\end{figure}


	A firing sequence $\sigma$ from $M$ to $M'$ in $\net$
	corresponds naturally to an integer flow $f$ on $G$,
	where for all vertices $p$ and $q$ 
	corresponding to places of the IO net,
	$f(i,p)=M(p), f(p,o)=M'(p)$ and $f(p,q)$ is equal to 
	the number of transitions from $p$ to $q$ in $\sigma$.
	This flow has value $|M|=|M'|$.
	
	Conversely, an integer flow
	of value $|M|=|M'|$
	corresponds to a firing sequence in $\net$,
	provided $\net$ has a fixed set of enabled token moves.
	Let us consider such a flow $f$.
	It corresponds to a multiset $\theta$ of token moves.
Starting with the marking $M$,
we remove from the multiset some token move 
with the source place  having  more tokens than in $M'$
and fire some corresponding enabled transition.
We continue until we reach $M'$.
The details of the construction
and its  correctness proof
are purely technical and can be found in the appendix.

We see that verifying a certificate
requires  a polynomial number of invocations 
of a polynomial-time algorithm.
This concludes the proof.
\end{proof}

In fact the reachability problem is \NP-complete.

\begin{restatable}{theorem}{ThNonforgetHard}
Reachability problem for non-forgetting IO nets is \NP-hard.
\end{restatable}
\begin{proof}[Sketch]     \NP-hardness of reachability is proved by a reduction from the \NP-complete SAT problem.
        Consider a SAT instance represented as a circuit of binary ``NAND'' ($\neg{(x \wedge y)}$) operations.
        One can construct a net such that 
        its runs correspond to the input nodes of the circuit 
        choosing
        arbitrary input values, and the operation nodes of the circuit evaluating
        the function given the chosen values of the inputs.
        The technical details are provided in the appendix.
\end{proof}

\section{Second restriction: token counts}
\label{sec:restriction2}

Another property of the \PSPACE-hardness reduction for IO nets
is the low number of tokens in each place.
Specifically, no reachable marking puts more than one token in any place. 
Some systems exhibit a very different behaviour.
For instance in most cases of chemical reaction networks, 
 the number of individual molecules is much larger 
than the number of species of molecules.
Additionally, we do not expect any chance ``near-misses'' 
between the configuration of the molecules before and after a reaction sequence.
If the total amount of molecules of some group of species
before the reaction sequence 
is
approximately equal 
to the amount of molecules of some other group of species afterwards,
there must be a precise equality following from some conservation laws.

This behaviour can be formalized by the following condition.

\begin{definition}
        A pair of markings $M$ and $M'$ of an IO net of place set $P$
        is a \emph{near-miss} pair
	if there exists sets of places $X$ and $Y$ such that
        $0 < |M(X)-M'(Y)| \le |P|^3$.
        A pair which  is not a near-miss is called a \emph{no-near-miss} pair.
\end{definition}


Observe that each place of markings $M$ and $M'$ 
such that $M,M'$ are a no-near-miss pair
can be either unmarked or contain at least
$|P|^3$ tokens. 
This can be seen by examining sets $X=\set{p}$ and $Y=\emptyset$, or $X=\emptyset$ and $Y=\set{p}$ in the definition.

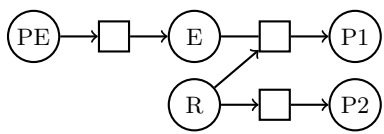
\begin{figure}
\vskip-5mm
        \centerline{\mbox{\begin{tikzpicture}[->, node distance=1cm, auto, thick]
\node[place](PE){PE};
\node[transition] (produce) [right=5mm of PE] {};
\node[place](E)[right=5mm of produce]{E};
\node[transition](use) [right=5mm of E] {};
\node[place](R)[below=2mm of E]{R};
\node[place](P1)[right=5mm of use]{P1};
\node[transition](avoid) [right=5mm of R] {};
\node[place](P2)[right =5mm of avoid]{P2};

\path[->]
(PE) edge (produce) (produce) edge (E)
(R) edge (use) (use) edge (P1)
(R) edge (avoid) (avoid) edge (P2)
;
\path[-]
(E) edge (use)
;
\end{tikzpicture}}}
        \caption{An example of an IO
        net with enzyme production and use.}
        \label{fig:dep-pathways}
\vskip-5mm
\end{figure}

\begin{example}
Consider the IO  net of Figure \ref{fig:dep-pathways}
which models a system
where an enzyme $E$ can be produced by an enzyme producer $PE$, and
where a resource molecule $R$ can transform 
 into a product molecule $P1$ in the presence of an enzyme $E$,
 or into a product molecule $P2$.
On the one hand,
the total amount of the two products $P1$ and $P2$
together must
match the amount of resource $R$ consumed;
on the other hand, it would be surprising if the two
products were produced in 
the same amounts with high but imperfect precision,
as there is nothing ensuring such an approximate equality.
Informally, we can consider 
the scales from an example of \cite{bistability}
cited in \cite{angeli2007petri}.
Five species of molecules are considered
        in a milliliter-scale cell (although with a different net
        which is not immediate observation).
The concentrations of molecules are measured in picomoles per milliliter.
As a picomole contains more than $10^{11}$ molecules,
equalities that hold up to $10^3$ molecules
have a relative error of $10^{-8}$.
Such equalities 
might be expected to follow from some
conservation laws and be precise.
\end{example}

\begin{theorem}
\label{thm:reach-near-miss}
The IO net reachability problem for no-near-miss pairs of markings is in P.
Moreover, there is a polynomial-time algorithm
such that for every pair  of markings $M,M'$
it either resolves reachability, giving a witness firing sequence if it exists,
or reports a near-miss in $M$ and $M'$.
\end{theorem}

Even though the no-near-miss property is \NP-complete
(e.g. via SUBSET-SUM),
making a proof of its violation
an alternative valid answer of the algorithm
simplifies IO reachability.

\begin{remark}
Requiring only that the initial and final markings of a firing sequence have
many tokens in the non-empty places does not give us a better complexity 
than the general \PSPACE-complete case.
\end{remark}

\begin{example}
Consider two markings on the net of Figure \ref{fig:dep-pathways},
        $M$
        with $200$ tokens in $PE$ and $400$ tokens in $R$,
        and 
        $M'$
        with $200$ tokens in $E$ and $400$ tokens in $P1$.
 The pair $(M,M')$ is a no-near-miss, and we will illustrate the algorithm 
 by verifying reachability from $M$ to $M'$.
\end{example}

The core idea of the algorithm 
is to maintain an increasing set of restrictions.
Once there are no restrictions to add,
we either construct a firing sequence from $M$ to $M'$
satisfying the obtained restrictions and no other ones,
use the restrictions to prove that $M$ cannot reach $M'$,
or find a near-miss in $M$ and $M'$.


\subsection{Restrictions}

We first recall some definitions from \cite{EsparzaRW19}, and then describe our restrictions and what it means for a restriction set to be stable.

\paragraph{Trajectories and histories.}
Since the transitions of IO nets do not create or destroy tokens, we can give tokens identities. 
Given a firing sequence, each token of the initial marking follows a \emph{trajectory} through the places of the net until it 
reaches the final marking of the sequence. The trajectories of the tokens between given source
and target markings constitute a \emph{history}. 

A \emph{trajectory} of  IO net $N$ is a sequence $\tau =p_1 \ldots p_k$ of places. We let  $\tau(i)$ denote the $i$-th place of $\tau$. The \emph{$i$-th step} of $\tau$ is the pair $\tau(i)\tau(i+1)$. A \emph{history} $H$ of length $h$ is a multiset of trajectories of length $h$. Given an index $1 \leq i \leq h$, \emph{the $i$-th marking of $H$}, denoted $M_{H}^i$, is defined as follows: for every place $p$, $M_{H}^i(p)$ is the number of trajectories $\tau \in H$ such that $\tau(i)=p$. The markings 
$M_{H}^1$ and $M_{H}^h$ are the \emph{initial} and \emph{final} markings of $H$, and we write
$M_{H}^1 \trans{H} M_{H}^h$. A history $H$ of length $h\geq 1$ is \emph{realizable} if there exist transitions $t_1, \ldots, t_{h-1}$ 
and numbers $k_1, \ldots, k_{h-1} \geq 0$ such that 
\begin{itemize}
\item $M_{H}^1 \trans{t_1^{k_1}}M_{H}^2 \cdots  M_{H}^{h-1} \trans{t_{h-1}^{k_{h-1}}} M_{H}^h$, where for every $t$  we define $\sourceMarking \trans{t^0} \targetMarking$ if{}f $\sourceMarking=\targetMarking$.
\item For every $1 \leq i \leq h-1$, there are exactly $k_i$ trajectories $\tau \in H$ such that $\tau(i)\tau(i+1) = p_s p_d$, where $p_s, p_d$ are the source and target places of $t_i$, and all other trajectories 
$\tau \in H$ satisfy $\tau(i)=\tau(i+1)$.
Moreover, there is at least one trajectory $\tau$ in $H$ such that $\tau(i)\tau(i+1) = p_o p_o$, where $p_o$ is the observed place of $t_i$.
                We say that $t_i$ \emph{realizes} step $i$ of $H$.
\end{itemize}

We say that $t_1^{k_1} \cdots t_{h-1}^{k_{h-1}}$ realizes $H$.
Intuitively, at a step of a realizable history only one transition occurs, although perhaps multiple times, for
different tokens.  From the definition of realizable history we immediately obtain:
\begin{itemize}
\item $\sourceMarking \trans{*} \targetMarking$ if{}f there exists a realizable history with $\sourceMarking$ and $\targetMarking$ as initial and final markings. 
\item Every firing sequence that realizes a history of length $h$ has accelerated length at most $h$.
\end{itemize}

\paragraph{Restriction definition.}
Given an IO net $\net$, places $p,q,r$ of $\net$, and 
 two markings $M$ and $M'$,
we say that \emph{a token goes from $p$ to $q$ via $r$} 
if there exists a realizable history $H$ of length $h$ between $M$ and $M'$
and
a trajectory $\tau$ in $H$ 
such that $\tau(1)=p$, $\tau(h)=q$ 
and $\tau(i)=r$ for some $i \in \set{1,\ldots,h}$.

Given a pair $M,M'$,
our algorithm computes a set $\restrictions$ of 
\emph{restrictions} of the form
$(p,r,q)$.
We say a restriction $(p,r,q)$ is \emph{correct}
if no token goes from $p$ to $q$ via $r$, 
\ie if
 there is no realizable history
from $M$ to $M'$ containing a trajectory from $p$ to $q$ passing through $r$.
We say that a pair of places $(p,q)$ is \emph{forbidden} if 
for all $r\in P$ the restriction $(p,r,q)$ 
is in $\restrictions$.
\emph{Forbidding} a pair $(p,q)$ means adding the restriction 
$(p,r,q)$ to $\restrictions$
for all $r\in P$.
A pair of places $(p,q)$ that is not forbidden is \emph{allowed}.


\paragraph{Flow graph.}
We define a correspondence between the
reachability problem in an IO net with a (correct) restriction set
and the maximum flow problem for a certain flow graph.

Let $\net$ be an IO net of place set $P$, 
let $M,M'$ be two markings of $\net$,
and let $\restrictions$ be a set of restrictions.
We define the \emph{flow graph} $G=(V,A,c)$
with $2|P|+2$ vertices.
There are two vertices for each place $p\in P$,
an ``initial'' copy $v_p^i$ and a ``final'' copy $v_p^f$,
as well as a distinguished inlet vertex $i$
and a  distinguished outlet vertex $o$.
For each place $p \in P$,
 there is an arc $a=(i, v_p^i)$ 
 with capacity $c(a)=M(p)$,
and an arc $a=(v_p^f,o)$ 
with  capacity $c(a)=M'(p)$.
For each pair of places $(p,q) \in P^2$
such that $(p,q)$ is allowed in $\restrictions$,
there is an arc $a=(v_p^i,v_q^f)$
from the initial $p$-labeled vertex to the final $q$-labeled vertex 
with infinite capacity.
Note that the maximum flow value in graph $G$ thus 
constructed is at most $|M|=|M'|$.

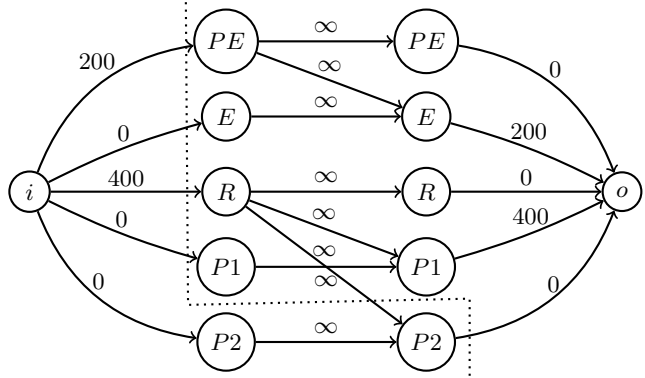
\begin{figure}[h]
\centering%

\begin{tikzpicture}[->, node distance=10mm, auto, thick]

\node[circle,draw] (inlet) {$i$};
\node[circle,draw] (p2a) [right=20mm of inlet] {$R$};
\node[circle,draw] (p1a) [above of=p2a] {$E$};
\node[circle,draw] (p0a) [above of=p1a] {$PE$};
\node[circle,draw] (p3a) [below of=p2a] {$P1$};
\node[circle,draw] (p4a) [below of=p3a] {$P2$};
\node[circle,draw] (p2b) [right=2cm of p2a] {$R$};
\node[circle,draw] (p1b) [above of=p2b] {$E$};
\node[circle,draw] (p0b) [above of=p1b] {$PE$};
\node[circle,draw] (p3b) [below of=p2b] {$P1$};
\node[circle,draw] (p4b) [below of=p3b] {$P2$};
\node[circle,draw] (outlet) [right=20mm of p2b] {$o$};

        \path[->,solid]
        (inlet) edge [bend left] node [above=+1.0mm] {200} (p0a)
        (inlet) edge [bend left=2mm] node [above=-0.3mm] {0} (p1a)
        (inlet) edge node [above=-0.5mm] {400} (p2a)
        (inlet) edge [bend right=2mm] node [above=-0.3mm] {0} (p3a)
        (inlet) edge [bend right] node [above=+0.0mm] {0} (p4a)
        (p0a) edge node [above] {$\infty$} (p0b)
        (p1a) edge node [above] {$\infty$} (p1b)
        (p2a) edge node [above] {$\infty$} (p2b)
        (p3a) edge node [below] {$\infty$} (p3b)
        (p4a) edge node [above] {$\infty$} (p4b)
        (p0a) edge node [above] {$\infty$} (p1b)
        (p2a) edge node [above] {$\infty$} (p3b)
        (p2a) edge node [above] {$\infty$} (p4b)
        (p0b) edge [bend left] node [above=+0.0mm] {0} (outlet)
        (p1b) edge [bend left=2mm] node [above=-0.0mm] {200} (outlet)
        (p2b) edge node [above=-0.3mm] {0} (outlet)
        (p3b) edge [bend right=2mm] node [above=-0.0mm] {400} (outlet)
        (p4b) edge [bend right] node [above=-0.5mm] {0} (outlet)
        ;

        \node[minimum size=0,inner sep=0] (cut1) [above left=3mm of p0a, yshift=0mm] {};
        \node[minimum size=0,inner sep=0] (cut2) [below left=3mm of p3a] {};
        \node[minimum size=0,inner sep=0] (cut5) [above right=4mm of p4b] {};
        \node[minimum size=0,inner sep=0] (cut6) [below right=4mm of p4b,yshift=1mm] {};

        \draw [-,dotted] 
        (cut1) -- (cut2) -- (cut5) -- (cut6)
        ;

\end{tikzpicture}
   \caption{Flow graph for the IO net of Figure \ref{fig:dep-pathways} with a cut.}%
   \label{fig:flow-double-cut}
\label{fig:flow-double}
\vskip-5mm
\end{figure}

\begin{example}
\label{ex:flow-double}
 Figure \ref{fig:flow-double} illustrates the flow graph $G$
 constructed for the IO net of Figure \ref{fig:dep-pathways},
 the markings $M=(200,0,400,0,0)$ and $M'=(0,200,0,400,0)$,
         and the restriction set 
         that allows only pairs
         of the form $(p,p)$ 
         and also the pairs $(PE,E)$,
         $(R,P1)$, $(R,P2)$.
\end{example}

A realizable history $H$ from $M$ to $M'$ naturally corresponds 
to a flow  of value $|M|$: the flow that saturates all the arcs with the finite capacities (i.e. the arcs from the inlet and to the outlet), and assigns to an infinite-capacity arc from $v_p^i$ to $v_q^f$ the number of trajectories from $p$ to $q$ in $H$. 
Since this flow saturates all the finite edges, it is a maximum flow.

\paragraph{Stable restriction set.}
 
We define the notion of a stable set of restrictions for a pair of marking $M$ and $M'$.
Intuitively, a stable set of restrictions does not immediately exclude
reachability from $M$ to $M'$,
and cannot be extended.

\begin{definition}
A set $\restrictions$ of correct restrictions for an IO net $\net$ and configurations $M$ and $M'$ is \emph{stable}
if the following conditions hold.
\begin{enumerate}
        \item The maximum flow in the corresponding flow graph is
                equal to the size $|M|$ of the configurations $M$ (and $M'$).
        \item For each two places $p$ and $q$,
                if there is a minimum cut of the flow graph with
                $v_p^i$ in the  outlet component
                and $v_q^f$ in the inlet component,
                the pair $(p,q)$ is forbidden.
        \item For each larger set of restrictions
$\restrictions'\supsetneq\restrictions$,
either there is a pair $(p,q)$
such that the triple $(p,p,q) \in \restrictions'\setminus\restrictions$,
or there is exists a transition $s\trans{o}d$ and triples $(p,s,q),(p',o,q')\notin\restrictions'$
and $(p,d,q)\in\restrictions'\setminus\restrictions$.
        \item For each larger set of restrictions
$\restrictions'\supsetneq\restrictions$,
either there is a pair $(p,q)$
such that the triple $(p,q,q) \in \restrictions'\setminus\restrictions$,
or there exist a transition $s\trans{o}d$ and
 triples $(p,d,q),(p',o,q')\notin\restrictions'$
and $(p,s,q)\in\restrictions'\setminus\restrictions$.
\end{enumerate}
\end{definition}

Each of these conditions prohibits some
property that can rule out reachability
or imply new  restrictions.
We give some intuition now, then prove formally in Section \ref{sec:firing-sequence}
that in the case where $M$ and $M'$ are a no-near-miss pair,
we can build a realizable history from $M$ to $M'$ from a stable set of restrictions. 
Moreover the history constructed will show that the set of restrictions cannot be extended.

We call the first two conditions
\emph{flow-based stability conditions}. 
The first condition corresponds to the fact that 
if a restriction set leads to a flow graph with a maximum
flow smaller than $|M|$, then there can be no realizable history
from $M$ to $M'$ consistent with such restrictions.
The second condition 
uses the fact that a minimum cut has the same value as a maximum flow,
which has size $|M|$ by the first flow-based condition.
Let $(p,q)$ be a pair violating the condition.
A max flow $f$ that uses the edge from $v_p^i$ to $v_q^f$
can be decomposed into a sum of two flows $f_1$ and $f_2$:
$f_1$ the flow with value 1 along path $i-v_p^i-v_q^f-o$ 
and $f_2=f-f_1$ which has value $|M|-1$.
Flow $f_1$ uses two arcs of the minimum 
cut thus yielding a contradiction by leaving
a cut of capacity $|M|-2$ to $f_2$.
This contradicts existence of a maximum flow 
using the edge from $v_p^i$ to $v_q^f$
and thus the existence of a realizable history
from $M$ to $M'$ with trajectories from $p$ to $q$.
%


\begin{example}
 Figure \ref{fig:flow-double-cut} illustrates a minimal cut on the flow graph $G$
 of Example \ref{ex:flow-double} in which
 the path $i\to{}v_{R}^i\to{}v_{P2}^f\to{}o$
 contains two arcs crossing the cut.
 The restriction set is not stable and $(R,P2)$ must be forbidden.
\end{example}

We call the last two conditions 
\emph{reachability-based stability conditions}.
They rule out an inductive proof
of a larger restriction set in the following sense. 
Given a larger set $\restrictions'$ which violates  one of these conditions, 
we will show by induction on the step number that any realizable history deduced from $\restrictions$ is also coherent with $\restrictions'$, and thus we can replace $\restrictions$ with the larger set $\restrictions'$.



\subsection{Firing sequence construction}
\label{sec:firing-sequence}

We show how to construct a firing sequence from a stable restriction set, 
possibly reporting a near-miss instead.
The proof that the near-miss reports are correct is after the construction, in Section \ref{sec:correctness}.

 Given a flow graph $G=(V,A,c)$, we define two operations
 on the capacity $c$  
 relative to a place pair $(p,q) \in P^2$ and an integer $k > 0$.
 \emph{Increasing $c$ by $k$ along $(p,q)$}
         consists in increasing $c(i,v_p^i)$ and $c(v_q^f,o)$ by $k$.
 \emph{Decreasing $c$ by $k$ along $(p,q)$}
         consists in decreasing $c(i,v_p^i)$ and $c(v_q^f,o)$ by $k$.
         This decreasing operation is not possible 
         if $c(i,v_p^i)$ or $c(v_q^f,o)$ are smaller than $k$.

\paragraph{From stable restriction set to solution flow.}
Given a stable set of restrictions with $b$ allowed pairs $(p,q)$,
 a \emph{solution flow}
is a result of the following procedure:
Construct the flow graph $G$.
Decrease the capacity by $|P|$ along each allowed pair;
if this step fails because some arc has insufficient capacity,
terminate the algorithm and report that 
$M,M'$ is a near-miss pair.
Otherwise, compute a maximal flow.
If it has value less than $|M|-b\times{}|P|$,
terminate the algorithm and report that
$M,M'$ is a near-miss pair.
Otherwise, increase its capacity by $|P|$ along each (allowed) pair.

\begin{example}
\label{ex:restrictions-solution-flow}
In our running example, consider a stable set of restrictions $\restrictions$
allowing only the triples 
$(PE,PE,E)$,
$(PE,E,E)$,
$(R,R,P1)$,
and
$(R,P1,P1)$.
This corresponds to a solution flow assigning the edges of the path
$i\to{}PE\to{}E\to{}o$ the value $200$
and the edges of the path 
$i\to{}R\to{}P1\to{}o$ the value $400$.
\end{example} 
        
Observe that 
when a solution flow exists, it might not be unique.
The algorithm builds a firing sequence from the solution flow.

\paragraph{From solution flow to firing sequence.}
Let
$\restrictions$ be a stable restriction set of the algorithm,
and let $f$ be a corresponding solution flow. 
Intuitively, our construction of the solution flow 
makes sure the flow has value
at least $|P|$ along each pair $(p,q)$ 
 allowed by $\restrictions$.
We use the  reachability-based stability conditions
to construct a realizable history from this flow,
 such that for every pair $(p, q)$
there are at most $f(v_p^i,v_q^f)$ trajectories from $p$ to $q$.

        We define three markings $M_m,M_i$ and $M_f$.
We denote $\Rset(p,q)$ the set of all places $r$
such that the triple $(p,r,q)$ is allowed, \ie $(p,r,q)\notin\restrictions$.
        Let $M_m$ be the marking such that $M_m(r)$
	is equal to the cardinality of the 
set $\set{(p,q) | r\in{}\Rset(p,q)}$ for all $r$.
        Let $M_i$ be the marking such that
        $M_i(p)=\sum_q |\Rset(p,q)|$.
	Note that as $|\Rset(p,q)|\leq{}|P|$
	we have $M_i(p)\leq{}f(i,v_p^i)$.
	Symmetrically, let $M_f$ be the marking such that
	$M_f(q)=\sum_p |\Rset(p,q)|$;
	we have $M_f(q)\leq{}f(v_q^f,o)$.
	We are going to construct a history from $M_i$ to $M_m$
	and from $M_m$ to $M_f$.

\begin{example}
        In our running example with $\restrictions$, 
        we obtain $M_i=\multiset{PE,PE,R,R}$,
        $M_f=\multiset{E,E,P1,P1}$, and $M_m=\multiset{PE,E,R,P1}$.
\end{example}

We build a history from $M_i$ to $M_m$ 
 with  trajectories
labeled by allowed pairs $(p,q)$ with many trajectories per pair.
Each trajectory for pair $(p,q)$ starts in place $p$.
The stability condition guarantees that we can extend some
trajectory to extend the set of places reached by trajectories
labeled $(p,q)$, until trajectories of every pair
have reached all allowed intermediate places $r$ such that $(p,r,q)$ is allowed.
For each reached place $r$
some trajectory stays in $r$ until the end of the history.
The history from $M_m$ to $M_f$ is built in a similar way
but using backward search from $M_f$.
After combining the two histories into a history from $M_i$ to $M_f$,
we duplicate some trajectory for each pair of places
until we have a history from $M$ to $M'$.
The construction consists
of technical details and can be found in the appendix.



Finally, we extract a firing sequence from the realizable history from $M$ to $M'$
by associating a transition and an iteration count to each step of the history.
Each step with $k$ trajectories
going from $p_s$ to $p_d$ with $p_s\neq p_d$
is associated to a transition $t$ iterated $k$ times from $p_s$ to $p_d$,
where $t$ realizes the step.
This is possible by realizability of the history.

\subsection{Correctness given a stable restriction set}
\label{sec:correctness}

 We  prove that given a 
 stable set of correct restrictions,
 the algorithm always yields a correct answer
in polynomial time.
 In case of a near-miss, both reporting the near-miss
 and correctly resolving reachability is considered a correct answer.

A near miss is reported in two cases of the solution flow construction, the second being more technical.
We give a sketch of the proof, the technical details are provided
in the appendix.

 \begin{restatable}{lemma}{LmNMReports}
 The near-miss reports are correct.
 \end{restatable}
\begin{proof}[Sketch]
We prove that the algorithm's reports of  near-misses are correct
for a net $\net$, markings $M,M'$ and a stable set of restrictions $\restrictions$.
A near miss is reported in two cases.
In the first case we cannot decrease some edge capacity by $|P|$,
after having attempted at most $|P|-1$ decreases for this edge beforehand.
This corresponds to a place of $M$ or $M'$ having more than $0$
but less than $|P|^2$ tokens,
which constitutes a near-miss.

In the second case after decreasing the capacity by $|P|$
along each of the $b$ allowed pairs,
there is some cut $(V_I,V_O)$ with capacity less than $|M|-b|P|$.
Each decrease operation decreases the capacity of each cut
at most by $2|P|$, so the original capacity of the
cut is less than $|M|+b|P|$.
On the other hand, it is strictly more than $|M|$,
as decreasing by $|P|$ along some pair reduced the capacity by more than $|P|$,
which is impossible for any minimum cut
by the second flow-based stability condition.
The
sets
 $X= V_I \cap \set{v_p^i | p\in P }$ and $Y= V_I \cap \set{v_p^f | p\in P }$
provide a near-miss.

\end{proof}

If the algorithm does not report a near-miss, 
then it successfully constructs
a solution flow and reports that $M'$ is reachable from $M$.
A realizable history can be constructed from the solution flow, proving that $M$ can reach $M'$.
Moreover the realizable history and then firing sequence from $M$ to $M'$ can be constructed in polynomial time and are correct by construction.
 
 \begin{restatable}{lemma}{LmPolyTime}
 The algorithm runs in polynomial time given a stable set of restrictions.
 \end{restatable}
 
The runtime analysis is straightforward, and can be found in the appendix. 

\subsection{Computing a stable restriction set}
\label{stable-restriction}

We show that there is a polynomial algorithm
that either computes a stable restriction set,
or correctly reports unreachability.
Starting with the empty set of restrictions,
the algorithm 
repeatedly finds violations of the stability conditions
and modifies the restriction set 
by adding some correct restrictions,
or reports unreachability.
Once no violations can be found, the algorithm terminates.
As the total number of possible triples is $|P|^3$,
only a polynomial number of iterations is needed.
It remains to show that the violations
as well as the corresponding additional correct restrictions
can be found in polynomial time.

\paragraph{First condition.} 
A violation can be found 
by computing the maximum flow.
Such a violation immediately implies
unreachability, since a realizable history induces a maximum flow of value $|M|$.

\paragraph{Second condition.} 
A violation can be found by
considering all the allowed pairs of places $(p,q)$
and computing the maximum flow after
decreasing the capacity by one
along $(p,q)$.
If the decrease is successful and the
maximum flow is $|M|-2$,
then $(p,q)$ is a violating pair, as argued in the section with the flow-based stability conditions.
We add new
correct restrictions by forbidding it.
If the decrease yields a maximum flow of $|M|-1$ 
then this pair does not create a violation.
If the decrease is not possible, then
we add new
correct restrictions by forbidding $(p,q)$.
Indeed if the decrease is not possible, then
the capacity between
$i$ and $v_p^i$ (resp. between $v_q^f$ and $o$) is zero.
The pair $(p,q)$ must be forbidden as there is no realizable history in which a token goes from $p$ to $q$.
The pair provides a violation of the 
condition by the cut which puts only $v_p^i$ and $o$ into the outlet component $V_O$ 
(resp., only $v_q^f$ and $i$ into the inlet component $V_I$)
and which is minimal because it has capacity $|M|$.

\paragraph{Third and fourth condition.} 
Checking for violations
of reachability-based stability conditions
shares part of the approach used to
construct a history out of a solution flow.
For the third condition, the algorithm enumerates
upper bounds on an extended set of restrictions
$\restrictions'$ violating the condition.
We start with $\restrictions'$ equal to all the triples.
We observe that $\restrictions'$ cannot contain $(p,p,q)$
for any pair $(p,q)$ 
such that $(p,p,q)$ is
not  in $\restrictions$. 
We exclude such $(p,p,q)$ from $\restrictions'$.
Then as long as there is
a transition $s\trans{o}d$ and
there are triples $(p,s,q),(p',o,q')\notin\restrictions'$
and $(p,d,q)\in\restrictions'\setminus\restrictions$,
we exclude $(p,d,q)$ from $\restrictions'$.
If we end up proving that $\restrictions'=\restrictions$,
there can be no violation.

Otherwise we prove that all the restrictions in $\restrictions'$
are correct and thus that $\restrictions$ is extendable.
Indeed, by induction, any history satisfying
the restrictions in  $\restrictions$ on all steps
must also satisfy the restrictions in $\restrictions'$.
%

The fourth condition is handled in a symmetric way.

\begin{example}
In our running example,
starting from an empty restriction set,
the second condition reports violations because decreasing is not possible.
It forbids all the pairs but $(PE,E)$,$ (PE,P1)$,$ (R,E)$,$ (R,P1)$.
Checking violations of the third condition
forbids all triples except 
$(PE,PE,E)$,
$(PE,E,E)$,
$(R,R,P1)$,
$(R,P1,P1)$,
$(R,P2,P1)$.
Checking the fourth condition additionally
forbids $(R,P2,P1)$
leaving only four allowed triples
$(PE,PE,E)$, $(PE,E,E)$, $(R,R,P1)$, $(R,P1,P1)$.
This set of restrictions is stable.
\end{example}


This procedure for constructing a stable set of restrictions,
coupled with the previous algorithm in which the stable set was part of the input,
completes the proof of Theorem \ref{thm:reach-near-miss}.

\section{Conclusion and future work}

We have considered two restrictions of 
the IO net reachability problem
with a promise for much simpler verification for some applications
and established the reachability complexity
in both these cases, which is \NP-complete in one case 
and polynomial in the other.
 
We leave the question of complexity of set-set reachability
under these restrictions for future research.
Another related question is
defining a notion of ``approximate'' reachability
that would provide a reduction in complexity
for IO nets,
as merely bounding the maximum difference between token counts
or the sum of differences preserves \PSPACE-hardness 
of the reachability problem.

\subsubsection*{Acknowledgements.}
We wish to thank Javier Esparza for useful discussions.
 We are also grateful to the anonymous reviewers
 for their advice regarding the presentation.

\bibliographystyle{plain}
\bibliography{references}

\appendix
\section{First restriction: transition enabling}

We provide the details of the firing sequence construction
out of a flow.

\begin{lemma}
An integer flow
	of value $|M|=|M'|$
	corresponds to a firing sequence in $\net$,
	provided $\net$ has a fixed set of enabled token moves.
\end{lemma}
\begin{proof}
	Let us consider such a flow $f$.
	It corresponds to a multiset $\theta$ of token moves 
	containing exactly $f(p,q)$ token moves from $p$
	to $q$ for every pair of places $p,q \in P$.
	To prove existence of a firing sequence for each such
	multiset,
	we consider the following (simple but inefficient)
	procedure, starting from $M$.
	We repeatedly pick a token move from some $p$ to some $q$ 
	 from the multiset such that $p$ has more tokens
	in the current marking than in the final marking $M'$.
	This is possible because IO nets are conservative: if there is no such place $p$ then the current marking is equal to $M'$ and we are done.
	We fire a transition of $\net$ with source place $p$ and destination place $q$,
	and remove the token move from the multiset.
	The existence of such a transition, enabled in the current marking,
	is given by the fact that the token move is enabled 
	and so there exists a transition of $\net$ from $p$ to $q$ whose observed place is marked.
\end{proof}

We
describe the reduction from SAT to the reachability problem for non-forgetting IO nets.

\ThNonforgetHard*
\begin{proof}
        \NP-hardness of reachability is proved by a reduction from the SAT problem.
        Consider a SAT instance represented as a circuit of binary ``NAND'' ($\neg{(x \wedge y)}$) operations 
        (any propositional formula can be converted into such form in linear time).
        We construct a net with the following places.
        \begin{itemize}
        \item For each input $x_i$ of the SAT circuit we add places $x_i^\bot$, $x_i^0$, $x_i^1$.
        Informally, marking these places corresponds to the input value
        being unknown, set to $0$ and to $1$ respectively.
        \item For each operation node $n_j$, we add places 
        $n_j^{(\bot,\bot)}$, $n_j^{(\bot,1)}$, $n_j^{(1,\bot)}$,
        $n_j^0$, $n_j^1$.
        Informally, these places correspond to our knowledge about the inputs
        and the output value of the node $n_j$:
        we can know neither input, know that one of the inputs is $1$,
        or know the output value of the node being $0$ or $1$
        (if one output is $0$, the node has the value $1$
        regardless of the other input).
        \end{itemize}
        
        The transitions are as follows.
        \begin{itemize}
        \item A token can move from a place $x_i^\bot$ to
        either of the places $x_i^0$ or $x_i^1$.
        \item A token in one of the places 
        $n_j^{(\bot,\bot)}$, $n_j^{(\bot,1)}$, $n_j^{(1,\bot)}$
        can observe a token in 
        $p_k^0$ or $p_k^1$ where $p_k$ is an input to $n_j$
        and move to the place corresponding to its updated information
        about the arguments.
        \item Let $n_o$ be the output operation node. 
        Any token can observe a token in $n_o^1$ and 
	perform any move that would be allowed by some observation
	(ensuring the non-forgetting property),
        or move to $n_o^1$.
        \end{itemize}

        The initial marking puts one token into each $x_i^\bot$ and $n_j^{(\bot,\bot)}$.

        Such a net is a non-forgetting IO net, 
        and it is easy to see that any execution in this net from the initial marking
        corresponds to guessing some inputs and evaluating the circuit.
        In particular, the marking with all the tokens in $n_o^1$ is reachable
        iff the circuit is satisfiable. This completes the proof.
\end{proof}

\section{Second restriction: token counts}

Below are the omitted or sketched proofs for the polynomial algorithm for reachability of no-near-miss pairs.

\subsection{From solution flow to firing sequence}

First we provide the details of the construction 
of a history from a solution flow.

We start with the history from $M_i$ to $M_m$.
We first produce an ordering of the triples 
$(p,r,q)$ not in $\restrictions$ and not of the form $(p,p,q)$,
and associate a transition to each of them using the first reachability-based stability condition satisfied by our stable set $\restrictions$.
We initialize $\restrictions'$  to be the set of triples 
$(p,r,q)$ not in $\restrictions$ and not of the form $(p,p,q)$.
Note that the first reachability-based stability condition
ensures that for each allowed pair $(p,q)$,
the triple $(p,p,q)$ is allowed.
Indeed, a restriction set additionally 
forbidding the pair $(p,q)$ violates the condition.
While $\restrictions\ne\restrictions'$,
we pick 
a transition $s\trans{o}d$ and
triples $(p,s,q),(p',o,q')\notin\restrictions'$
and $(p,d,q)\in\restrictions'\setminus\restrictions$.
We number $(p,d,q)$, associate to it the transition $s\trans{o}d$,
remove it from $\restrictions'$ and continue.


We say
	a place $r'$ is an \emph{initially-reachable child}
	of place $r$ for pair $(p,q)$
if $(p,r',q)$ was excluded from $\restrictions'$
	because of  some transition $r\xrightarrow{s}r'$.
	The notion of \emph{initially-reachable descendant}
	is defined by transitive and reflexive closure over the initially-reachable child relation.

	We define the first step of the history from $M_i$ to $M_m$ to
	 consist of trajectories of length $1$ 
	such that there is exactly one trajectory in $p$
	for each triple $(p,r,q)$ such that $r\in\Rset(p,q)$.
	We label each trajectory with its triple $(p,r,q)$.
	This first step corresponds to the marking $M_i$.
	The idea is to extend each trajectory of $M_i$  labeled $(p,r,q)$ 
	from $p$ until it reaches place $r$.

We construct the history
by adding one step per triple in our ordering.
At each new step $i+1$, we maintain two things:
\begin{itemize}
\item If there is a trajectory $\tau$ with $\tau(i)=p$ 
then there is a trajectory $\tau'$ with $\tau'(i+1)=p$,
\ie a place once marked by the history stays marked.
\item If $\hat{r}$ is
	 the last place of a $(p,r,q)$-labeled trajectory,
	 then 
	$r$ is an initially-reachable descendant of place $\hat{r}$ for pair $(p,q)$,
and $(p,\hat{r},q)$ is the triple with the largest number in the ordering such that this holds. 
\end{itemize}
Initially this holds as $p$ is an ancestor for all $r\in\Rset(p,q)$.

At each step, we pick the next triple $(p,r',q)$ in the ordering.
It is associated to a transition $\hat{r} \xrightarrow{s}r'$.
	For every place $d$ which is a descendant of $r'$,
	we extend trajectories labeled $(p,d,q)$ with a step from $\hat{r}$ to $r'$.
	The rest of the  trajectories in the history
	are  extended with ``horizontal'' steps
	preserving their  current places.
By construction,
for some $p',q'$ the triple
$(p',s,q')$ is earlier in the certificate,
	so the history includes a trajectory having already
reached the place $s$ and still in $s$,
and so realizability is preserved.
	Eventually all  the trajectories
	reach the place $r$ of their label $(p,r,q)$.
As a trajectory marked with $(p,r,q)$ 
reaches $r$ and stays there afterwards,
the final marking puts in each place $r$
exactly $\set{(p,q) \mid r\in{}\Rset(p,q)}$,
thus we reach the  marking $M_m$.

	We construct a realizable history from
	$M_m$ to $M_f$ in a symmetrical way.
	We  produce an ordering of the triples 
$(p,r,q)$ not in $\restrictions$ and not of the form $(p,q,q)$,
and associate a transition to each of them using the second reachability-based stability condition satisfied by our stable set $\restrictions$.
We define the symmetric notions 
 of  \emph{finally-reachable child} and \emph{finally-reachable descendant}.
 Then we construct the trajectories of the history from $M_m$ to $M_f$,
 working backwards from $M_f$ on  trajectories 
labeled $(p,r,q)$ from $q$ until $r$.
	
	We concatenate these two histories
	(identifying the trajectories
	labeled $(p,r,q)$ in them)
	to obtain a history from $M_i$
	to $M_f$
	with 
	$|\Rset(p,q)|\leq|P|\leq{}f(v_p^i,v_q^f)$
	trajectories from $p$ to $q$.
	We pick an arbitrary trajectory 
	from $p$ to $q$ and 
        increase its multiplicity in the multiset by
	$f(v_p^i,v_q^f)-|\Rset(p,q)|$.
	We do this until there are $f(v_p^i,v_q^f)$ trajectories 
	for every pair of places $(p,q)$.
	This provides a realizable history from $M$ to $M'$.
Realizability is preserved as the sets of steps 
at each position in the history stay the same 
and only multiplicities change.
Such changes cannot create a violation
of the realizability criterion.

\begin{example}
\begin{figure}
        \centerline{\mbox{\includegraphics[]{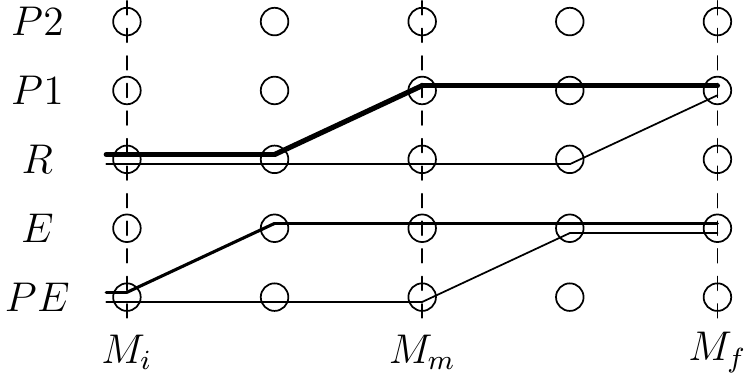}}}
        \caption{A history obtained from a solution flow and a stable set of restrictions.
        Bold trajectories are taken with multiplicities $199$ and $399$.}
        \label{fig:solution-history}
\end{figure}
In our running example, 
from the previously shown restrictions and solution flow in Example \ref{ex:restrictions-solution-flow},
we can obtain the history illustrated in Figure \ref{fig:solution-history}
with $199$ copies of trajectory $PE,E,E,E,E$,
$1$ copy of $PE,PE,PE,E,E$,
399 copies of $R,R,P1,P1,P1$,
and 1 copy of $R,R,R,R,P1$.
Note that this history results from a certain ordering, and that a different ordering provides a different history.
\end{example}

\subsection{Correctness given a stable restriction set}

\LmNMReports*
\begin{proof}

A near miss is reported in two cases.
In the first case, the report arises because
 decreasing capacity $c$ of flow graph
 $G=(V,A,c)$ by $|P|$ along the $b$ allowed pairs of $\restrictions$ is impossible.
In this case,
$M,M'$ is a near-miss pair 
as 
there are less than $|P|^2$ tokens in some marked place of $M$ or $M'$.
This can be seen by examining sets $X=\set{p}$ and $Y=\emptyset$, or $X=\emptyset$ and $Y=\set{p}$ in the definition of a near-miss.

In the second case, the report arrises because decreasing capacity $c$ of flow graph $G=(V,A,c)$ by $|P|$ along the $b$ allowed pairs of $\restrictions$
leads to a maximum flow value less than $|M|-b\times{}|P|$.
 We call $c'$ the capacity post-decrease, and note $G'=(V,A,c')$.
Equality of the minimum cut and the maximum flow
 gives existence of a cut in $G'$ with capacity less than $|M|-b\times{}|P|$.
 Consider such a cut $(V_I, V_O)$ of capacity $\kappa'< |M|-b\times{}|P|$.
 We write $\kappa$ the capacity of cut $(V_I, V_O)$ in $G$
 before the decrease operation.
 Since the maximum flow, and thus minimum cut,
 of $G$ is $|M|$,
 we have $\kappa \ge |M|$.
 Therefore there exists an allowed pair $(p,q)$ such that
 the arcs $(i, v_p^i)$ and $(v_q^f,o)$ both cross the cut,
 as otherwise $\kappa' \ge |M|-b\times{}|P|$.
 Since the restriction set is stable,
 decreasing by $1$ along any allowed pair keeps any cut capacity in $G$
 bigger or equal to $|M|-1$. 
 Thus we have $\kappa > |M|$.
 By structure of $G$ and $G'$, 
 the decreasing operation can reduce 
 a cut capacity by at most $2b\times{}|P|$.
 So $\kappa -\kappa' \le 2b|P|$, and using the inequalities above
 as well as the fact that there are at most $b \le |P|^2$ allowed pairs,
 we get $|M| <\kappa < |M|+|P|^3$. 
 
 Consider the following two vertex sets based on cut $(V_I, V_O)$.
 Let $X= V_I \cap \set{v_p^i | p\in P }$ and $Y= V_I \cap \set{v_p^f | p\in P }$.
 Our cut is finite, so only finite capacity arcs cross it,
 namely the arcs
 from the inlet to vertices $v_p^i$ and from vertices $v_p^f$ to the outlet.
 The capacity in $G$ of this cut is thus $\kappa=M(P\setminus{X})+M'(Y)$.
 Since $|M| < \kappa < |M|+|P|^3$ and $|M|=M(P)$,
 we know $0<M(P\setminus{}X)+M'(Y)-M(P)<|P|^3$.
 By set considerations
 $ M(P)-M(P\setminus{}X) = M(X)$,
 and so finally $0<M'(Y)-M(X)<|P|^3$.
 The sets $X,Y$ prove that $M,M'$ are a near-miss.
\end{proof}


\LmPolyTime*
\begin{proof}
First the algorithm computes a stable set of restrictions.
To this end it repeatedly finds violations
of stability conditions and deduces additional restrictions.

A check of flow-based stability conditions
requires a computation of maximum flow
in the flow graph corresponding to the current 
restriction set,
then one additional maximum flow computation for 
each allowed pair.
A check of reachability-based stability conditions
can be performed
by repeated enumeration
of possible combinations of three triples
and verification of existence of corresponding transitions.
It is clear that both checks can be implemented in polynomial time.

Each iteration either terminates the algorithm
or adds at least one new triple to the set of
known correct restrictions.
As the total number of triples
is polynomial 
and each iteration takes polynomial time,
the total runtime of computing a stable set is polynomial.

If a stable set of restrictions is found,
a solution flow can be found by a maximum flow
algorithm, unless a near-miss is reported.

If a near-miss is reported, a proof can be constructed
either directly by checking all the token counts,
or by running a minimum cut algorithm.

If a solution flow is found, a history constructed
contains two steps per allowed triple, 
one in $M_i$ to $M_m$ and one in $M_m$ to $M_f$.
The numbering of triples for each part
can be built by enumerating combinations of three triples,
then a pass through the numbering is enough to build
reachability child relations.
One more traversal of the numbering,
adding one step to each trajectory at each step,
is enough to build the half-history.

To construct a firing sequence it suffices
to enumerate all pairs of horizontal and non-horizontal steps
at each position in the history,
and check all the transitions for each pair.
Note that identical steps of different trajectories need not
be considered separately.

We observe that all the steps can be performed in polynomial time.
\end{proof}

\end{document}